\documentclass[a4paper]{article}
\usepackage{graphicx}
\usepackage{amsmath,amsthm,amssymb,amsbsy}
\usepackage{mathrsfs}
\usepackage[affil-it]{authblk}
\usepackage{hyperref}
\usepackage{caption}
\usepackage{subcaption}
\usepackage{cite}
\usepackage{braket}
\usepackage[titletoc,title]{appendix}
\title{No-faster-than-light-signaling implies linear evolutions. A re-derivation}

\author{Angelo Bassi\footnote{\texttt{bassi@ts.infn.it}}}
\affil{Department of Physics, University of Trieste, Strada Costiera 11, 34151 Trieste, Italy}
\affil{Istituto Nazionale di Fisica Nucleare, Trieste Section, Via Valerio 2, 34127 Trieste, Italy}
\author{Kasra Hejazi\footnote{\texttt{hejazi\_kasra@physics.sharif.edu}}}
\affil{Department of Physics, Sharif University of Technology, P.O. Box 11155-9161, Tehran, Iran}

\date{}
\begin{document}
\maketitle

\begin{abstract}
There is a growing interest, both from the theoretical as well as experimental side, to test the validity of the quantum superposition principle, and of theories which explicitly violate it by adding nonlinear terms to the Schr\"odinger equation. We review the original argument elaborated by N. Gisin~\cite{gisin}, which shows that the non-superluminal-signaling condition implies that the dynamics of the density matrix must be linear. This places very strong constraints on the permissible modifications of the Schr\"odinger equation, since they have to give rise, at the statistical level, to a linear evolution for the density matrix. The derivation is done in a heuristic way here and is appropriate for the students familiar with the Textbook Quantum Mechanics and the language of density matrices.
 \end{abstract}

\section{Introduction}

Quantum Mechanics is a very successful physical theory. It has been revolutionary in many diverse areas, from particle physics to condensed matter and statistical physics, explaining many features of nature in a unified framework. 

Quantum Mechanics has introduced many new concepts to physics. Historically it came into play when Max Planck first assumed that the energy of the standing waves in a cavity can only take discrete values. Then he could explain the frequency distribution of the black-body-radiation. As in this example, one of the very first applications of the theory was to make quantities, especially energy, discrete; from which the theory gets its name. 

Many other concepts have been introduced by Quantum Theory. As a second example consider the idea of identical particles and indistinguishability; this is crucial in many ways to particle statistics in Statistical Physics. For instance, one cannot track down one electron, among many, and talk about its definite state, an idea which was almost impossible in non-Quantum way of thinking. Just to name some other examples from the large list of new concepts due to Quantum Theory, consider basic and at the same time important notions like spin, tunneling and the peculiar non-locality.

Besides new concepts, Quantum Mechanics has great power in predicting quantitative results; think of the Anomalous Magnetic Moment of the electron. It can be derived using the formalism of the next revolution, Quantum Field Theory, which was invented because of the need for a Relativistic Quantum Theory. The agreement between the theory and experiment is astonishing: more than ten decimal places.

%end 

Nevertheless, over the years several people repeatedly suggested to modify the Schr\"odinger equation, which is the fundamental dynamical law for quantum systems. One can identify three main reasons to do so. 

The first and most widespread reason is the so-called measurement process. Quantum mechanics predicts that there are physical states corresponding to the superposition of two other states. Such superpositions are evident experimentally in the microscopic world, but are absent in our everyday experience of the macro-world. One does not see anything in a superposition of eigenstates of position like: $[|\text{I am in Tehran}\rangle + |\text{I am in Trieste}\rangle]/\sqrt{2}$. This is especially important when dealing with macroscopic systems interacting with microscopic ones, i.e.~in measurements: in such cases the dynamics as given by the Schr\"odinger equation  immediately predicts the formation of macroscopic superpositions corresponding to two or more possible outcomes of the measurement.

To avoid superposition in macroscopic physics, the standard approach is to assume that a collapse of the state vectors takes place \cite{WeinQ,sakurai,shankar}.  As a result one has a theory which deals with the physical world in two different ways. At the very fundamental microscopic world everything evolves according to the Schr\"odinger equation, but at the stage of macroscopic interactions with the micro-world, the wave function collapses. Even the border between these two kinds of evolutions is ambiguously defined. Moreover, the interaction between the measurement apparatus and the system can be tracked down to the fundamental interactions between the particles of the apparatus and the system, which are described by the linear Schr\"odinger equation. But it is not possible to derive the collapse from such a scheme.\cite{DRM,measure}

There are several ways one can get out of this trouble. One example is \textit{Bohmian Mechanics} \cite{Bohma,Bohmb,BohmHiley,Bub, DurrGold,Durretal,Durrt,Holland}, which essentially does not contain collapse. Despite having a wave function, particles have definite locations in this theory. A measurement is simply an interaction between the measurement apparatus and the system, without any mystery. 
Another way out, which is a very natural thing to do, is to modify Schr\"odinger equation in order to have one dynamical evolution both for the micro- and the macro-world. The new equation should, approximately, correspond to the old Schr\"odinger evolution for microscopic systems, and should yield the collapse of the wave function in the macroscopic domain. Collapse Models are examples of these kinds of modifications \cite{DRM,grw,csl,diosi,sc1,rmp,Weinberg}. The dynamics is modified by adding non-linear and stochastic terms to the usual Schr\"odinger equation.

A second motivation for modifying the standard quantum dynamics comes from Quantum Field Theory. This is a linear theory (again, the superposition principle) and it is quite natural to think that this is only the first order approximation of a deeper level non-linear theory. The paradigmatic example is Newtonian gravity, which now we understand as the weak field limit of General Relativity. S. Weinberg first suggested such a possibility \cite{weinl}.

As a third motivation, there is the longstanding problem of the unification of quantum and gravitational phenomena. While waiting for a fully consistent and successful unified theory, one can write phenomenological equations. One such equations is the so-called Schr\"odinger-Newton equation, which is non-linear \cite{sn1,sn2,sn3}. Different parts of the wave function of a system interact among themselves through the Newtonian potential. 

All these efforts raise an important question: `Are there limits in modifying the Schr\"odinger equation?' The answer is yes. As Gisin convincingly showed~\cite{gisin}, {\it whichever way one modifies the Schr\"odinger equation for the wave function, the time evolution for the density matrix has to be linear if one wants to keep no-faster-than-light-signaling.} Here we review this argument with the purpose of giving a pedagogical introduction to the subject.

\section{Evolution of the ensembles}

We work in a standard quantum framework. The {\it states} of  any given physical system are described by the elements $|\psi\rangle$ of a Hilbert space $\mathscr{H}$. We want also to give a statistical description of the system, for which it is convenient to use an {\it ensemble (or mixture) of states.} One can identify any such ensemble with a set of state vectors and the corresponding probabilities, i.e.~$\{ d_i;|\psi_i\rangle \}$. The idea is that the system is in any of such states $|\psi_i\rangle$, but we do not know which one (classical ignorance). We only know the probabilities $d_i$ for the system to be in the states $|\psi_i\rangle$. (figure \ref{ensemble}) 

\begin{figure}
        \centering
        \begin{subfigure}[h]{0.48\textwidth}
                \centerline{\includegraphics[scale=0.7]{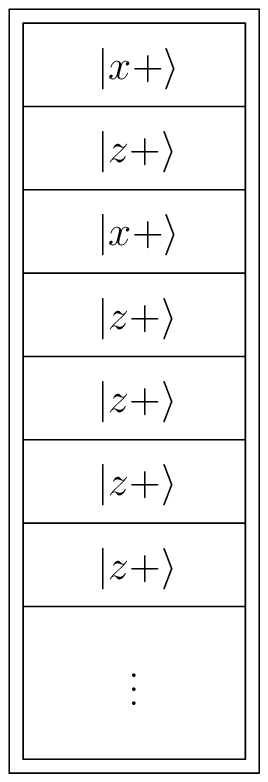}}
                \caption{}\label{ensemble_a}
        \end{subfigure}
        \begin{subfigure}[h]{0.48\textwidth}
                \centerline{\includegraphics[scale = 0.7]{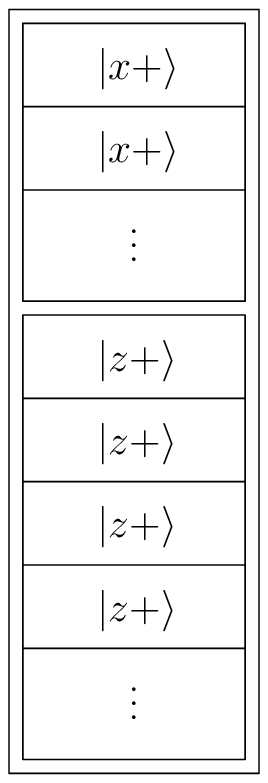}}
                \caption{}\label{ensemble_b}
        \end{subfigure}
      
        \caption{\label{ensemble} \footnotesize
         (\subref{ensemble_a}) A statistical ensemble, consisting of the states $|z+\rangle$ and $|x+\rangle = \tfrac{1}{\sqrt{2}} \left( |z+\rangle + |z-\rangle \right)$, with probabilities $2/3$ and $1/3$ respectively. Note that the population of the states in the ensemble is proportional to the probabilities.		
        (\subref{ensemble_b}) The same ensemble re-ordered. The ensemble now consists of two pure sub-ensembles.The second sub-ensemble is populated twice as the first sub-ensemble, and so, the probability of the state $|z+\rangle$ is twice as $|x+\rangle$.}
\end{figure}

The density matrix formalism is the appropriate tool for dealing with statistical mixtures. For a given mixture $\{ d_i;|\psi_i\rangle \}$, the {\it density matrix} is defined as $\rho \equiv  \sum_i  d_i |\psi_i\rangle\langle\psi_i|$. A density matrix is {\it pure} if there is only one element in the ensemble (with probability 1), and is {\it mixed} otherwise. 

We name the set of all density matrices as $\mathfrak{B}^+$, and the set of pure ones as $\mathfrak{B}^p$. 

\theoremstyle{definition}
\newtheorem*{mix}{Definition}

\begin{mix}
Two mixtures $\{d_i;|\psi_i\rangle \}$ and $\{ p_i;|\phi_i\rangle \}$ are equivalent if their corresponding density matrices are the same:
\begin{equation}
\sum_i  d_i |\psi_i\rangle\langle\psi_i| = \sum_i  p_i |\phi_i\rangle\langle\phi_i|.
\end{equation}
One can easily see that the equivalence of the mixtures is an equivalence relation.
\end{mix}

Let us consider a given dynamical evolution for the state vectors of a given system, which in principle has nothing to do with the Schr\"odinger equation; in particular, it might be non-linear\footnote{A generic non-linear operator acting on $|\psi\rangle$ will be shown as $O(|\psi\rangle)$, while for a linear operator we will use the notation $O|\psi\rangle$. Note that this will not be done for operators acting on density matrices, these operators in this paper are either linear or just acting on pure states. Moreover, we assume that the evolution of any initial state depends only on the state itself and not on the particular way the state is prepared.\cite{cavalcanti}}: 
\begin{equation}\label{1}
|\psi_t \rangle=\mathscr{T}_{(t,t_0)}(|\psi_{t_0}\rangle).
\end{equation}

This dynamical law automatically (and trivially) defines a dynamics also for {\it pure} density matrices: 
\begin{equation}\label{16}
\rho_{t_0}=|\psi_{t_0}\rangle\langle\psi_{t_0}| \to \rho_t=|\psi_t\rangle\langle\psi_t|,
\end{equation}
where $|\psi_t\rangle$ is given by equation \eqref{1}. So we have the following map defined on the space $\mathfrak{B}^p$ of pure density matrices from time $t_0$ to time $t$:
\begin{align} \label{2}
\mathscr{E}_{(t,t_0)} :\mathfrak{B}^p\to \mathfrak{B}^p\\  \nonumber
\rho_t = \mathscr{E}_{(t,t_0)} (\rho_{t_0}).
\end{align}

\begin{figure}[!h]
\centerline{\includegraphics[scale =0.7]{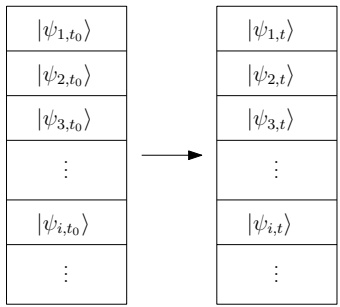}}
\caption{\footnotesize\label{ensemble_evolve}
The evolution of an ensemble; each single state vector evolves independently from the others. Therefore, each pure sub-ensemble evolves according to~\eqref{2}. Note that the statistical weights, being proportional to the number of the copies of each state vector, remain constant.}
\end{figure}

With the dynamics of pure ensembles known, there remains the question of how {\it mixed} ensembles evolve. One expects each of the states in the mixed ensemble to evolve according to evolution \eqref{1}, and independent of the other states. (look at figure~\ref{ensemble_evolve}) In other words, different pure sub-ensembles of an ensemble do not interact, i.e.~each one evolves independently according to equation \eqref{2}. Therefore a generic density matrix evolves in the following way:
\begin{align}\label{17}
\sum_i d_i |\psi_{i,t_0}\rangle\langle\psi_{i,t_0}| & \to  \sum_i d_i |\psi_{i,t}\rangle\langle\psi_{i,t}|  \\ \nonumber
&=\sum_i d_i \left[\mathscr{E}_{(t,t_0)}(|\psi_{i,t_0}\rangle\langle\psi_{i,t_0}|)\right],  \nonumber
\end{align}

This suggests the possibility of extending the map  $\mathscr{E}_{(t,t_0)}$ from the space of pure density matrices to the space of all density matrices:
\begin{align}\label{15}
\mathscr{E}_{(t,t_0)} :\mathfrak{B}^+\to \mathfrak{B}^+\\  \nonumber
\rho_t=\mathscr{E}_{(t,t_0)}(\rho_{t_0}).
\end{align}
However, in general this is not possible, because two different decompositions of the same initial density matrix $\rho_{t_0}$ may correspond to different density matrices after evolution. (figure~\ref{equivalent_ensembles}) In other words one density matrix can have different destinies according to the ensemble it is representing. This means that the map in~\eqref{15} is not well defined in general.

\begin{figure}
        \centering
        \begin{subfigure}[!h]{0.48\textwidth}
                \centerline{\includegraphics[scale = .7]{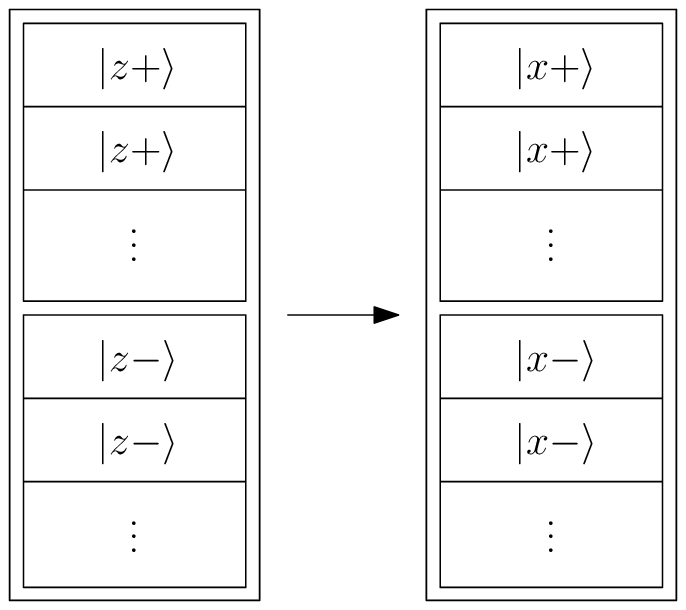}}
                \caption{}\label{evolve_z}
        \end{subfigure}
        \begin{subfigure}[!h]{0.48\textwidth}
                \centerline{\includegraphics[scale = .7]{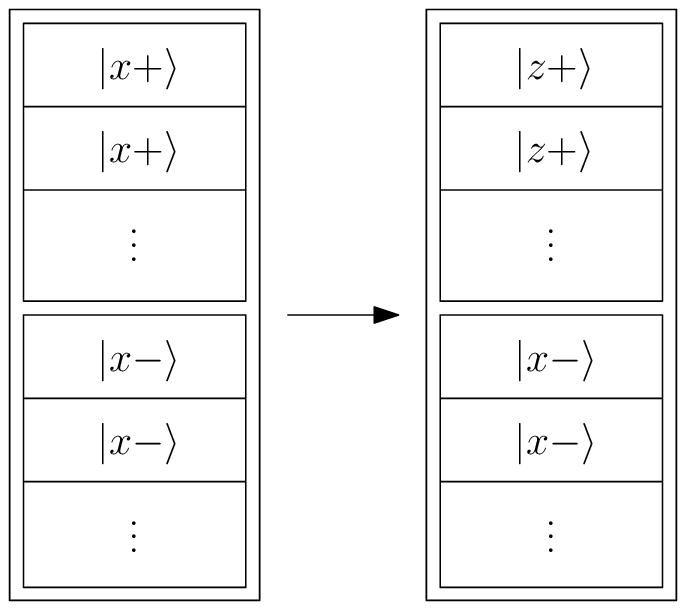}}
                \caption{}\label{evolve_x}
        \end{subfigure}
      
        \caption{\label{equivalent_ensembles} \footnotesize
         (\subref{ensemble_a}) An ensemble of states $\ket{z+}$ and $\ket{z-}$, with $1/2$ probabilities. The density matrix for this ensemble is $\rho^a_{t_0} = \frac{1}{2} \ket{z+}\bra{z+} + \frac{1}{2} \ket{z-}\bra{z-} = \frac{1}{2} I$. Each state evolves according to evolution~\eqref{1}. We assume that the fate of $\ket{z+}$ is $\ket{x+}$, and that the fate of $\ket{z-}$ is $\ket{x-} = \frac{1}{\sqrt{2}}(\ket{z+}+ \ket{z-})$. Hence, the final mixture consists of $\ket{x+}$ and $\ket{x-}$ with probability $1/2$ for each. The final density matrix is  $\rho^a_{t} = \frac{1}{2} \ket{x+}\bra{x+} + \frac{1}{2} \ket{x-}\bra{x-} = \frac{1}{2} I$. \\
        (\subref{ensemble_b}) An equivalent mixture, consisting of  $\ket{x+}$ and $\ket{x-}$ with equal probabilities ($1/2$). The initial density matrix is $\rho^b_{t_0} = \frac{1}{2} \ket{x+}\bra{x+} + \frac{1}{2} \ket{x-}\bra{x-} = \frac{1}{2} I$. Let us assume that the states evolve as: $\ket{x+} \to \ket{z+}$ and $\ket{x-} \to \ket{x-}$. The latter happens due to a non-linearity in the presumed evolution. The final density matrix will be $\rho^b_{t} = \frac{1}{2} \ket{z+} \bra{z+} + \frac{1}{2} \ket{x-}\bra{x-} \neq \frac{1}{2} I$, which is clearly different from $\rho_t^a$. This is an example of two initially equivalent ensembles, becoming inequivalent, due to the non-linear evolution of their state vectors.}
        
\end{figure}

However, Gisin's theorem states that if one wants to rule out the possibility for superluminal communication, different mixtures which are equivalent have to remain equivalent even after the evolution. So that the map $\mathscr{E}_{(t,t_0)}$ in~\eqref{2} can be really extended to the  whole set of density matrices as in~\eqref{15}. The theorem will be proved in the next section.

One can also think of stochastic evolutions of the state vectors, which is the case in Collapse Models, and see what constraints can be put on this type of the evolution. Gisin's theorem immediately generalizes also to such evolutions. (see Appendix~\ref{A1} for details.)

\section{The theorem}

We will be working with finite-dimensional Hilbert spaces. That is reasonable, because in case of infinite-dimensional spaces one can always work with a finite-dimensional subspace. The finite-dimensional Hilbert space describing the state vectors of the system we will be working with, is $\mathscr{H}$. 

We will use un-normalized vectors, which will make our work easier and our results  simpler to read. So we will change notation. The \textit{square norm of a vector} from now on will show its probability in the ensemble, so that instead of the mixture  $\{ d_i;|\psi_i\rangle \}$ we will be using  $\{|\psi_i\rangle \}$, and:
$$\langle\psi_i|\psi_i\rangle=d_i.$$

We will be dealing with the time evolution of two arbitrary but equivalent mixtures $\{|\psi_i\rangle \}$ and $\{|\phi_j\rangle \}$ throughout, where $i=1,\ldots,n_{\psi}$ ; $j=1,\ldots,n_{\phi}$. Because the two ensembles are equivalent we know that one density matrix $\rho$ describes both, i.e.~$\rho=\sum_i |\psi_i\rangle\langle\psi_i|=\sum_i |\phi_i\rangle\langle\phi_i|$. Now we have all the ingredients to formulate Gisin's argument.

As in a usual Bell-type experiment, there are Alice and Bob, far apart from each other, who make measurements on entangled particles\footnote{For simplicity we will speak of particles, but the two types of systems on which Alice and Bob make measurements can be generic.}.
We use the Hilbert space $\mathscr{H}$ to describe the states of the particles traveling towards Alice's place and use the Hilbert space $\mathscr{K}$ for the ones traveling towards Bob's region. Suppose $\mathscr{K}$ is large enough for the argument to apply. It can be shown that the following state $|V\rangle$ can be shared between the two parties:

\begin{equation}\label{24}
|V\rangle=\sum_i |\psi_i\rangle \otimes |\alpha_i\rangle =\sum_i |\phi_i\rangle \otimes |\beta_i\rangle,
\end{equation}
where $\{|\psi_i\rangle \}$ and $\{|\phi_j\rangle \}$ are the two equivalent mixtures in $\mathscr{H}$ that, we started with, upon which  Alice will make measurements, and $\{|\alpha_i\rangle\}$ and $\{|\beta_i\rangle\}$ are two different orthonormal bases of $\mathscr{K}$ (Bob's system). The preparation of such a state is possible for any arbitrary choice of the equivalent mixtures $\{|\psi_i\rangle \}$ and $\{|\phi_j\rangle \}$. (see Appendix B for a proof and details.)
\bigskip

\begin{figure}[!h]
\centerline{\includegraphics[scale = 0.8]{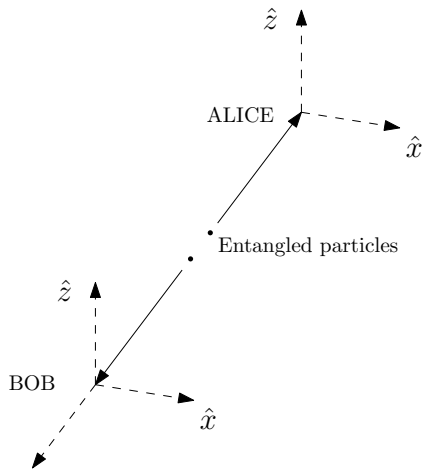}}
\caption{\footnotesize\label{AliceBob}
An ensemble of entangled states $\ket{V}$ is shared between Alice and Bob. Take the state $\ket{V} = \frac{1}{\sqrt{2}} \left( \ket{x+}\otimes\ket{z+} + \ket{x-}\otimes\ket{z-} \right) $ as an example. Only the spin components of the entangled particles have been shown in the state. Suppose Bob decides to measure $S_z$ on his particles. Measuring the $z$-component spin of each particle by Bob, makes the state of its entangled particle in Alice's system definite, due to a collapse: either $\ket{x+}$ if Bob gets $+1$ or $\ket{x-}$ if Bob gets $-1$ in his measurement. So a mixture of states $\ket{x+}$ and $\ket{x-}$ with equal probabilities is prepared for Alice.\\
What if Bob decides to measure $S_x$? One can easily see that the state $\ket{V}$ can be represented as $\ket{V}= \frac{1}{\sqrt{2}} \left( \ket{z+}\otimes\ket{x+} + \ket{z-}\otimes\ket{x-} \right)$ also. Therefore, this time a mixture of $\ket{z+}$ and $\ket{z-}$ with equal probabilities will be prepared. Note that these two kinds of mixtures are equivalent, i.e.~:
$\frac{1}{2} \left\{ \ket{x+}\bra{x+} + \ket{x-}\bra{x-} \right\} = \frac{1}{2} \left\{ \ket{z+}\bra{z+} + \ket{z-}\bra{z-} \right\}.$
}
\end{figure}

Suppose an ensemble of states $|V\rangle$ is shared between Alice and Bob. Bob has two different choices (among the many); either to measure the observable $O_{\alpha}$ which has the  $|\alpha_i\rangle$'s as eigenvectors,  or to measure the observable $O_{\beta}$ which has the $|\beta_i\rangle$'s as  eigenvectors. The first type of measurement  will prepare the ensemble of states $\{|\psi_i\rangle\}$ and the second will prepare the ensemble $\{|\phi_i\rangle\}$ for Alice, as a result of the collapse of the state vector. They correspond to different mixtures, but the density matrix and so all of the expectation values Alice can measure are the same. The two ensembles are equivalent. (figure~\ref{AliceBob}) Notice that Bob can prepare any equivalent ensemble for Alice, by measuring a properly chosen observable, without any need to change the previously shared state respectively. (look at Appendix \ref{A2})
Therefore, this first part of the argument shows that {\it different but equivalent mixtures can be prepared at a distance, by using suitable entangled states and measurements.}

Now the crucial point comes. Suppose these two mixtures $\{|\psi_i\rangle \}$ and $\{|\phi_j\rangle \}$, after some time, become inequivalent, because of the dynamics given by Eq.~\eqref{17}. Then Alice can find out the difference between them by making appropriate measurements and computing averages. Therefore she is in the position to understand which observable Bob decided to measure, even if Bob is arbitrarily far apart.  There is the possibility for superluminal communication, no matter how long it takes for the two mixtures to become appreciably different from each other. Therefore, if we hold on the assumption that there cannot be faster-than-light signaling, {\it different but equivalent  mixtures have to stay equivalent while time passes.} 

The above argument implies that  the evolution map $\mathscr{E}_{(t,t_0)}$  defined in~\eqref{2} can be extended to the whole space of density matrices as in~\eqref{15}, using the rule in~\eqref{17}.
{\it Linearity} of the map is a consequence of using statistical nature of ensembles, i.e.~pure sub-ensembles evolve independently. Suppose $\rho$ is a convex sum of $\rho_1$ and $\rho_2$:
\begin{equation*}
\rho=\lambda\rho_1+(1-\lambda)\rho_2.
\end{equation*}
If a possible ensemble for $\rho_1$ is $\{|a_i\rangle\}$ and one for $\rho_2$ is $\{|b_i\rangle\}$,\footnote{Remember that the norm of the vectors shows their weight in the ensemble.} then one possible ensemble for $\rho$ will be $\{\sqrt{\lambda}|a_i\rangle,\sqrt{1-\lambda}|b_i\rangle\}$. The time evolution for $\rho$ is given by (assuming the extended map~\eqref{15}):

\begin{eqnarray}
\mathscr{E}_{(t,t_0)}\rho & = & \mathscr{E}_{(t,t_0)}\Big(\sum_i \lambda |a_i\rangle\langle a_i|+\sum_i (1-\lambda) |b_i\rangle\langle b_i|\Big) \\
& = & \sum_i \lambda \mathscr{E}_{(t,t_0)}\Big(|a_i\rangle\langle a_i|\Big)+\sum_i (1-\lambda) \mathscr{E}_{(t,t_0)}\Big(|b_i\rangle\langle b_i|\Big)\\
& = &  \lambda \mathscr{E}_{(t,t_0)}\Big(\sum_i  |a_i\rangle\langle a_i|\Big)+(1-\lambda) \mathscr{E}_{(t,t_0)}\Big(\sum_i  |b_i\rangle\langle b_i|\Big)\\ 
& = & \lambda\mathscr{E}_{(t,t_0)}\rho_1+(1-\lambda)\mathscr{E}_{(t,t_0)}\rho_2.
\end{eqnarray}
In the fist line, we simply re-wrote $\rho$ in terms of the statistical mixture; in going from the first to the second line, we used~\eqref{17} applied to $\rho$; in going from the second to the third line, we used again~\eqref{17}, this time applied to the two mixtures $\sum_i |a_i\rangle\langle a_i|$ and $\sum_i |b_i\rangle\langle b_i|$; in going from the third to the forth line, we use the definition of $\rho_1$ and $\rho_2$ in terms of the mixtures defining them.

Now that we know that the evolution of the density matrix is linear, we can make yet further progress by understanding how the state vectors evolve. As is discussed thoroughly in \cite{GhirardiGrassi}, the only deterministic evolutions of the state vectors, which induce linear evolutions at the level of density matrices, are the linear unitary evolutions of the state vectors, i.e.~ the old Schr\"odinger evolutions of the state vectors. As \cite{GhirardiGrassi} further states, in order to obtain a non-Schr\"odinger evolution at the level of state vectors, consistent with no-Signaling, one has to let the state vectors evolve stochastically. We will talk more about this in the Conclusion.

\subsection*{Conclusion}
As stated in the introduction there are attempts for modifying the time evolution in Quantum Mechanics in order to solve the measurement problem. Gisin's argument restricts the possibilities of modifying the evolution equations up to the border of linearity: \textit{The evolution for the density matrix has to be linear}, no matter how the state vector's evolution is, e.g.~deterministic, stochastic, etc. 

According to the work of G. Lindblad \cite{Lin} and of V. Gorini, A. Kossakowsky and E.C.G. Sudarshan \cite{GKS}, a linear evolution for the density matrix can be only of the Lindblad type, if one adds the two further requirements of a quantum-dynamical-semigroup type of equation, and complete positivity:
\begin{equation}\label{Lindblad2}
 \frac{\text d\rho_t}{\text dt} = -i[H,\rho_t] + \sum_{k=1}^n\left(L_k\rho_tL_k^{\dagger} - \frac12 L_k^{\dagger}L_k\rho_t - \frac12 \rho_tL_k^{\dagger}L_k\right).
\end{equation}
The first assumption amounts to requiring a Markovian evolution, which is a reasonable starting point for writing down a dynamical equation (all fundamental equations in physics, thus far, are Markovian; non-Markovian processes become relevant at the statistical level, when phenomenological equations are the only reasonable way for analyzing a complex system). The second assumption becomes necessary when entangled systems are taken into account.

Then, according to \cite{BDH} the only possible way to modify the Schr\"odinger equation for the wave function, which leads to a Lindblad type of equation for the density matrix, is the way collapse models do it. Therefore, given the above premises, one can conclude that {\it the one given by collapse models is the only possible way to modify the Schr\"odinger equation, if one requires no superluminal communication}.

\subsection*{Acknowledgement}
A.B. wishes to thank  the EU project NANOQUESTFIT,  INFN and  the Templeton Foundation project `Experimental and Theoretical Exploration of Fundamental Limits of Quantum Mechanics' for partial financial support. He also  acknowledges support from the COST Action MP1006.
K.H. wants to thank Theoretical Physics Department of Trieste University for hospitality and also INFN for financial support. We both thank Professor Vahid Karimipour for introducing K.H.~to the research group of A.B.~.

\begin{appendix}

\section{}\label{A1}

In this appendix, the possibly stochastic evolutions of the state vectors are discussed.

 In the case of stochastic evolution of a state vector, initial pure states evolve into mixed mixtures. Think of two identical state vectors present in the ensemble, they will in general evolve into different states. One cannot tell what the future of each state is, but can talk about the statistics of the future of identical copies; something which can be determined based on the type of the stochastic evolution involved. (look at figure~\ref{stochastic_evolution}) One expects a pure density matrix $\rho_{t_0}$ to evolve as:
 $$\rho_{t_0}=|\psi_{t_0}\rangle\langle\psi_{t_0}| \to \rho_t= {\mathbb E}\left[|\psi_t\rangle\langle\psi_t|\right],$$
where ${\mathbb E}[\cdot]$ denotes the stochastic average and $\ket{\psi_t}$ has a degree of randomness, despite the fact that $\ket{\psi_{t_0}}$ is known. Accordingly, the map in Eq.~\eqref{2} changes into: 
\begin{align} \label{2b}
\mathscr{E}_{(t,t_0)} :\mathfrak{B}^p\to \mathfrak{B}^+\\  \nonumber
\rho_t =\mathscr{E}_{(t,t_0)} (\rho_{t_0}).
\end{align}

 Now each pure sub-ensemble in a mixture evolves as~\eqref{2b}. Gisin's argument follows.

\bigskip
\begin{figure}[!h]
\centerline{\includegraphics[scale = 0.7]{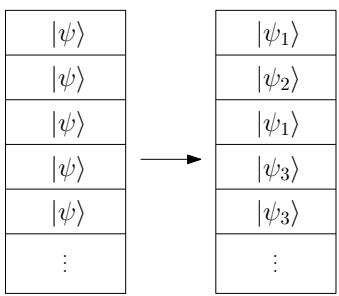}}
\caption{\footnotesize\label{stochastic_evolution}
A pure ensemble getting mixed due to the stochastic evolution. The relative population, again shows the statistical weight in the final ensemble.
}
\end{figure}

\section{}\label{A2}

In this appendix it is shown that a state like the one in \eqref{24} can be shared between the two parties:
\begin{equation*}
|V\rangle=\sum_i |\psi_i\rangle \otimes |\alpha_i\rangle =\sum_i |\phi_i\rangle \otimes |\beta_i\rangle,
\end{equation*}
where  $\{|\psi_i\rangle \}$ and $\{|\phi_j\rangle \}$ are two given equivalent mixtures ($i=1,\ldots,n_{\psi}$ ; $j=1,\ldots,n_{\phi}$.), and there is no further restriction on them. 
Let us start at a basic level and prove this via some lemmas.
The lemmas are presented in a similar fashion to the results of \cite{hughston}. 

We will work in a finite-dimensional space. Any density matrix has a spectral decomposition, i.e.~:  
\begin{equation} \label{20}
\rho=\sum_i |\chi_i\rangle\langle\chi_i|,
\end{equation}
with $ |\chi_i\rangle$'s being orthogonal:
\begin{equation*}
\langle\chi_i|\chi_j\rangle=\lambda_i \delta_{ij} , \qquad 0 < \lambda_i \leq 1.
\end{equation*}

\theoremstyle{plain}

\newtheorem{lem}{Lemma}

\begin{lem}\label{3}
Suppose $\rho$ is the density matrix associated to the ensemble $\{|\psi_i\rangle \}$, then any single vector $|\psi_i\rangle $ can be written as a linear combination of the vectors in the spectral decomposition~\eqref{20} of $\rho$.
\end{lem}
\begin{proof}
Take $\{|\chi_1\rangle,\ldots,|\chi_n\rangle\}$ as the spectral decomposition of $\rho$. If they are a basis, the theorem is trivial. Suppose they are not. Expand the set of $\{|\chi_1\rangle,\ldots,|\chi_n\rangle\}$ to an orthogonal basis of the whole Hilbert space by adding $\{|\chi_{n+1}\rangle,\ldots,|\chi_{n'}\rangle\}; n< n'$ to it. Note that this can always be done. 

Any vector $|\psi_i\rangle$ can be written as a linear combination of these basis-vectors: 
\begin{equation}\label{21}
|\psi_i\rangle=\sum_{j=1}^{n'} m_{ij}|\chi_j\rangle,
\end{equation}
where $m_{ij}$'s are complex numbers. Take $j$ such that $n<j\leq n'$.
By the orthogonality of $|\chi_i\rangle$'s:

\begin{align}
\langle\chi_{j}|\rho|\chi_{j}\rangle=\sum_{i=1}^{n} \langle\chi_{j}|\chi_i\rangle\langle\chi_i|\chi_{j}\rangle=0,
\end{align}
but also:

\begin{align}
\langle\chi_{j}|\rho|\chi_{j}\rangle=0 &=\sum_{i=1}^{n_\psi} \langle\chi_{j}|\psi_i\rangle\langle\psi_i|\chi_{j}\rangle \\ \nonumber
&=\sum_{i=1}^{n_\psi} m_{ij} m^*_{ij} \\ \nonumber
&=\sum_{i=1}^{n_\psi} |m_{ij}|^2  .  \\ \nonumber
\end{align}
A sum of non-negative numbers is equal to zero, i.e.~they are all zero. So for every $i$:

\begin{equation}
m_{ij}=0 ; \qquad n<j\leq n'.
\end{equation}
Hence, one can rewrite equation~\eqref{21} as:
\begin{equation}\label{22}
|\psi_i\rangle=\sum_{j=1}^{n} m_{ij}|\chi_j\rangle .
\end{equation}

\end{proof}

\newtheorem{lem2}[lem]{Lemma}

\begin{lem2}\label{5} 
Suppose that the density matrix $\rho$ has the spectral decomposition $\{|\chi_1\rangle,\ldots,|\chi_n\rangle\}$. A set of vectors $\{|\psi_i\rangle \}$ has the same density matrix if and only if 
\begin{equation}\label{8}
|\psi_i\rangle=\sum_{j=1}^n m_{ij} |\chi_j\rangle,
\end{equation} 
with $m_{ij}$ satisfying:
\begin{equation}\label{9}
\sum_{i=1}^{n_\psi} m_{ij}m^*_{ik}=\delta_{jk}.
\end{equation}

\end{lem2}
\begin{proof}
First assume that density matrix of the mixture is $\rho$, so:
\begin{equation}\label{6}
\sum_{i=1}^{n_\psi} |\psi_i\rangle\langle\psi_i|=\sum_{l=1}^n |\chi_l\rangle\langle\chi_l|.
\end{equation}
Equation \eqref{22} takes the form:
$$|\psi_i\rangle=\sum_{j=1}^n m_{ij}|\chi_j\rangle,$$
multiplying by $\langle\chi_k|$ we have:
\begin{equation}\label{7}
\langle\chi_k|\psi_i\rangle=m_{ik}\lambda_k.
\end{equation}
Then multiplying equation \eqref{6} by $\langle\chi_j|$ from left and by $|\chi_k\rangle$ from right we have:
\begin{align}
\sum_i \langle\chi_j|\psi_i\rangle \langle\psi_i|\chi_k\rangle &= \sum_l \lambda_j \delta_{jl} \lambda_k \delta_{kl} \\ \nonumber
&= \lambda_k^2 \delta_{jk}.
\end{align}
And by using \eqref{7} one immediately gets:
$$\sum_i m_{ij}m^*_{ik}=\delta_{jk}.$$
Conversely suppose \eqref{8} and \eqref{9} hold, so we have:
\begin{align*}
\sum_i |\psi_i\rangle\langle\psi_i|  &=\sum_{i,j,k} m_{ij}m_{ik}^* |\chi_k\rangle\langle\chi_j| \\
& =\sum_j |\chi_j\rangle\langle\chi_j|\\
&= \quad \rho .
\end{align*}
 \end{proof}

Consider the columns of $\boldsymbol{m}$ (the matrix having $m_{ij}$ as components, $i$ showing the row and $j$ the column). They are $n$ vectors in $\mathbb{C}^{n_{\psi}}$. According to~\eqref{9} they are orthonormal and so linearly independent. So $n\leq n_{\psi}$ and also the number of columns of $\boldsymbol{m}$ is less than or equal to the number of its rows.

We will need a square matrix instead of the possibly rectangular one above. For that, one can add $n_{\psi}-n$ null vectors to the set of $\{|\chi_1\rangle,\ldots,|\chi_n\rangle\}$ to make them $\{|\chi_1\rangle,\ldots,|\chi_n\rangle,|0\rangle,\ldots,|0\rangle\}$ and add columns to $\boldsymbol{m}$ to make them an orthonormal basis of $\mathbb{C}^{n_{\psi}}$ (the columns were already $n$ orthonormal vectors, look at~\eqref{9}). So the expanded $\boldsymbol{m}$ will be an $n_{\psi}\times n_{\psi}$ square matrix and we will name it $\boldsymbol{M}$. It is easy to see that $\boldsymbol{M}$ is unitary~\eqref{9}. In the matrix notation we will have:
\begin{equation}\label{10}
\begin{pmatrix}
|\psi_1\rangle  \\
|\psi_2\rangle \\
\vdots \\
|\psi_{n_{\psi}}\rangle
\end{pmatrix}
= \quad \boldsymbol{M} \quad
\begin{pmatrix}
|\chi_1\rangle  \\
\vdots \\
|\chi_n\rangle  \\
|0\rangle \\
\vdots \\
|0\rangle
\end{pmatrix}.
\end{equation}
Note that $\boldsymbol{M}$ acts only on the arrays and is not a Hilbert space operator. 

$\boldsymbol{M}$ can be expanded even more. In fact, one can make a square matrix of any size satisfying~\eqref{9} in this way by adding enough zeros to $|\psi_i\rangle$'s and $|\chi_i\rangle$'s and adding components to $\boldsymbol{M}$ without changing the old ones, and making sure it remains unitary.

\newtheorem{lem3}[lem]{Lemma}

\begin{lem3}\label{14}
Consider two sets of vecotrs $\{|\psi_i\rangle \}$ and $\{|\phi_i\rangle \}$ all belonging to the Hilbert space $\mathscr{H}$ and assume that they correspond to the same density matrix, i.e.~$\rho=\sum_i |\psi_i\rangle\langle\psi_i|=\sum_i |\phi_i\rangle\langle\phi_i|$. There will be a Hilbert space $\mathscr{K}$ which has the following property: There is a vector in the tensor product space $\mathscr{H}\otimes\mathscr{K}$ which can be written in this way:
\begin{equation*}
|V\rangle=\sum_i |\psi_i\rangle \otimes |\alpha_i\rangle =\sum_i |\phi_i\rangle \otimes |\beta_i\rangle,
\end{equation*}
where $|\alpha_i\rangle$'s are orthonormal and $|\beta_i\rangle$'s are orthonormal. 
\end{lem3}
\begin{proof}
Let $\{|\chi_1\rangle,\ldots,|\chi_n\rangle\}$ be the spectral decomposition of the density matrix $\rho$. Assume $n_{\psi}\leq n_{\phi}$ without loss of generality. Now as stated before, one can add zeros, if necessary, to the sets of vectors to make the number of members of every set $n_{\phi}$. So from \eqref{10} we have:
\begin{equation*}
|\psi_i\rangle=\sum_{j=1}^{n_{\phi}} M_{ij}|\chi_j\rangle,
\end{equation*}
$\boldsymbol{M}$ is unitary, so:
\begin{equation}\label{11}
|\chi_i\rangle=\sum_{j=1}^{n_{\phi}} M^{\dagger}_{ij}|\psi_j\rangle.
\end{equation}

Let $\mathscr{K}$ be $\mathbb{C}^{n_{\phi}}$, and let $\{|a_i\rangle\}$ be any orthonormal basis for $\mathscr{K}$. Consider the following vector:
\begin{align}\label{12}
|V\rangle &=\sum_i |\chi_i\rangle \otimes |a_i\rangle \\ \nonumber
& =\sum_{i,j} M^{\dagger}_{ij}|\psi_j\rangle \otimes |a_i\rangle \\ \nonumber
& =\sum_j |\psi_j\rangle \otimes |\alpha_j\rangle, \\  \nonumber
\end{align}
where we defined $|\alpha_j\rangle=\sum_i M^{\dagger}_{ij}|a_i\rangle$. Let us see if $|\alpha_i\rangle$'s are orthonormal:
$$\langle\alpha_i|\alpha_j\rangle =\sum_k (M^{\dagger}_{ik})^*M^{\dagger}_{jk}=\sum_k M_{ki}M^*_{kj}=\delta_{ij}.$$

We can do the same thing for $|\phi_i\rangle$'s:
\begin{equation}\label{13}
|V\rangle =\sum_i |\chi_i\rangle \otimes |a_i\rangle=\sum_j |\phi_j\rangle \otimes |\beta_j\rangle,
\end{equation}
with
$$\langle\beta_i|\beta_j\rangle=\delta_{ij}.$$

So from \eqref{12} and \eqref{13}:
\begin{equation}\label{23}
|V\rangle=\sum_i |\psi_i\rangle \otimes |\alpha_i\rangle =\sum_i |\phi_i\rangle \otimes |\beta_i\rangle.
\end{equation}
\end{proof}
 \end{appendix}
 
This concludes the mathematical background necessary to prove Gisin's theorem.

\end{document}